\documentclass[journal]{IEEEtran}
\usepackage{array}
\usepackage{arydshln}
\usepackage[T1]{fontenc}
\usepackage[utf8]{inputenc}
\usepackage{mathtools}

\usepackage{amssymb,mathrsfs}
\usepackage{amsthm}
\usepackage{bm}
\usepackage{scalerel}
\usepackage{nicefrac}
\usepackage{microtype} 
\usepackage[shortlabels]{enumitem}
\usepackage{graphicx}
\usepackage{epstopdf}
\DeclareGraphicsExtensions{.eps,.png,.jpg,.pdf}

\usepackage{url}
\usepackage{colortbl}
\usepackage{booktabs}
\usepackage{multirow}
\usepackage{colortbl,xcolor}
\usepackage{soul}
\usepackage{xparse,xstring}
\usepackage{calc}
\usepackage{etoolbox}

\makeatletter
\@ifpackageloaded{natbib}{
	\relax
}{
	\usepackage{cite}
}
\makeatother

\usepackage{array}
\newcolumntype{L}[1]{>{\raggedright\let\newline\\\arraybackslash\hspace{0pt}}m{#1}}
\newcolumntype{C}[1]{>{\centering\let\newline\\\arraybackslash\hspace{0pt}}m{#1}}
\newcolumntype{R}[1]{>{\raggedleft\let\newline\\\arraybackslash\hspace{0pt}}m{#1}}

\makeatletter
\let\MYcaption\@makecaption
\makeatother
\usepackage[font=footnotesize]{subcaption}
\makeatletter
\let\@makecaption\MYcaption
\makeatother

\usepackage{glossaries}
\makeatletter
\sfcode`\.1006

\let\oldgls\gls
\let\oldglspl\glspl

\newcommand\fussy@ifnextchar[3]{%
	\let\reserved@d=#1%
	\def\reserved@a{#2}%
	\def\reserved@b{#3}%
	\futurelet\@let@token\fussy@ifnch}
\def\fussy@ifnch{%
	\ifx\@let@token\reserved@d
		\let\reserved@c\reserved@a
	\else
		\let\reserved@c\reserved@b
	\fi
	\reserved@c}

\renewcommand{\gls}[1]{%
\oldgls{#1}\fussy@ifnextchar.{\@checkperiod}{\@}}
\renewcommand{\glspl}[1]{%
\oldglspl{#1}\fussy@ifnextchar.{\@checkperiod}{\@}}

\newcommand{\@checkperiod}[1]{%
	\ifnum\sfcode`\.=\spacefactor\else#1\fi
}

\robustify{\gls}
\robustify{\glspl}
\makeatother

\newacronym{wrt}{w.r.t.}{with respect to}
\newacronym{RHS}{R.H.S.}{right-hand side}
\newacronym{LHS}{L.H.S.}{left-hand side}
\newacronym{iid}{i.i.d.}{independent and identically distributed}
\newacronym{SOTA}{SOTA}{state-of-the-art}

\usepackage{float}

\ifx\notloadhyperref\undefined
	\ifx\loadbibentry\undefined
		\usepackage[hidelinks,hypertexnames=false]{hyperref}
	\else
		\usepackage{bibentry}
		\makeatletter\let\saved@bibitem\@bibitem\makeatother
		\usepackage[hidelinks,hypertexnames=false]{hyperref}
		\makeatletter\let\@bibitem\saved@bibitem\makeatother
	\fi
\else
	\ifx\loadbibentry\undefined
		\relax
	\else
		\usepackage{bibentry}
	\fi
\fi

\usepackage{crossreftools}
\ifx\notloadhyperref\undefined
\else
	\relax
\fi

\usepackage{algorithm}
\usepackage{algpseudocode}

\ifx\loadbreqn\undefined
	\relax
\else
	\usepackage{breqn}
\fi

\makeatletter
\def\cleartheorem#1{%
    \expandafter\let\csname#1\endcsname\relax
    \expandafter\let\csname c@#1\endcsname\relax
}
\def\clearthms#1{ \@for\tname:=#1\do{\cleartheorem\tname} }
\makeatother

\ifx\renewtheorem\undefined
	\ifx\useTheoremCounter\undefined
		\newtheorem{Theorem}{Theorem}
		\newtheorem{Corollary}{Corollary}
		\newtheorem{Proposition}{Proposition}
		
	\else
		\newtheorem{Theorem}{Theorem}

	\fi

	\newtheorem{Definition}{Definition}
	\newtheorem{Example}{Example}

\fi

\theoremstyle{remark}

\theoremstyle{plain}

\newcommand{\qednew}{\nobreak \ifvmode \relax \else
		\ifdim\lastskip<1.5em \hskip-\lastskip
			\hskip1.5em plus0em minus0.5em \fi \nobreak
		\vrule height0.75em width0.5em depth0.25em\fi}



\NewDocumentCommand{\movedownsub}{e{^_}}{%
	\IfNoValueTF{#1}{%
		\IfNoValueF{#2}{^{}}
	}{%
		^{#1}
	}%
	\IfNoValueF{#2}{_{#2}}
}

\let\latexchi\chi
\RenewDocumentCommand{\chi}{}{\latexchi\movedownsub}

\newcommand{\calN}{\mathcal{N}}

\newcommand{\calP}{\mathcal{P}}

\newcommand{\calT}{\mathcal{T}}

\newcommand{\bY}{\mathbf{Y}}

\newcommand{\bbR}{\mathbb{R}}
\newcommand{\bbS}{\mathbb{S}}

\newcommand{\scN}{\mathscr{N}}

\DeclareSymbolFont{bsfletters}{OT1}{cmss}{bx}{n}
\DeclareSymbolFont{ssfletters}{OT1}{cmss}{m}{n}
\DeclareMathSymbol{\bsfGamma}{0}{bsfletters}{'000}
\DeclareMathSymbol{\ssfGamma}{0}{ssfletters}{'000}
\DeclareMathSymbol{\bsfDelta}{0}{bsfletters}{'001}
\DeclareMathSymbol{\ssfDelta}{0}{ssfletters}{'001}
\DeclareMathSymbol{\bsfTheta}{0}{bsfletters}{'002}
\DeclareMathSymbol{\ssfTheta}{0}{ssfletters}{'002}
\DeclareMathSymbol{\bsfLambda}{0}{bsfletters}{'003}
\DeclareMathSymbol{\ssfLambda}{0}{ssfletters}{'003}
\DeclareMathSymbol{\bsfXi}{0}{bsfletters}{'004}
\DeclareMathSymbol{\ssfXi}{0}{ssfletters}{'004}
\DeclareMathSymbol{\bsfPi}{0}{bsfletters}{'005}
\DeclareMathSymbol{\ssfPi}{0}{ssfletters}{'005}
\DeclareMathSymbol{\bsfSigma}{0}{bsfletters}{'006}
\DeclareMathSymbol{\ssfSigma}{0}{ssfletters}{'006}
\DeclareMathSymbol{\bsfUpsilon}{0}{bsfletters}{'007}
\DeclareMathSymbol{\ssfUpsilon}{0}{ssfletters}{'007}
\DeclareMathSymbol{\bsfPhi}{0}{bsfletters}{'010}
\DeclareMathSymbol{\ssfPhi}{0}{ssfletters}{'010}
\DeclareMathSymbol{\bsfPsi}{0}{bsfletters}{'011}
\DeclareMathSymbol{\ssfPsi}{0}{ssfletters}{'011}
\DeclareMathSymbol{\bsfOmega}{0}{bsfletters}{'012}
\DeclareMathSymbol{\ssfOmega}{0}{ssfletters}{'012}

\makeatletter
\newcommand*\rel@kern[1]{\kern#1\dimexpr\macc@kerna}
\newcommand*\widebar[1]{%
  \begingroup
  \def\mathaccent##1##2{%
    \rel@kern{0.8}%
    \overline{\rel@kern{-0.8}\macc@nucleus\rel@kern{0.2}}%
    \rel@kern{-0.2}%
  }%
  \macc@depth\@ne
  \let\math@bgroup\@empty \let\math@egroup\macc@set@skewchar
  \mathsurround\z@ \frozen@everymath{\mathgroup\macc@group\relax}%
  \macc@set@skewchar\relax
  \let\mathaccentV\macc@nested@a
  \macc@nested@a\relax111{#1}%
  \endgroup
}
\makeatother

\DeclareMathOperator{\var}{var}

\DeclareMathOperator{\cov}{cov}

\DeclareMathOperator{\ima}{im}

\newcommand{\ifbcdot}[1]{\ifblank{#1}{\cdot}{#1}}

\DeclarePairedDelimiterX\abs[1]{\lvert}{\rvert}{\ifbcdot{#1}}
\DeclarePairedDelimiterX\parens[1]{(}{)}{\ifbcdot{#1}}
\DeclarePairedDelimiterX\brk[1]{[}{]}{\ifbcdot{#1}}
\DeclarePairedDelimiterX\braces[1]{\{}{\}}{\ifbcdot{#1}}
\DeclarePairedDelimiterX\angles[1]{\langle}{\rangle}{\ifblank{#1}{\cdot,\cdot}{#1}}
\DeclarePairedDelimiterX\ip[2]{\langle}{\rangle}{\ifbcdot{#1},\ifbcdot{#2}}
\DeclarePairedDelimiterX\norm[1]{\lVert}{\rVert}{\ifbcdot{#1}}
\DeclarePairedDelimiterX\ceil[1]{\lceil}{\rceil}{\ifbcdot{#1}}
\DeclarePairedDelimiterX\floor[1]{\lfloor}{\rfloor}{\ifbcdot{#1}}

\DeclareFontFamily{U}{matha}{\hyphenchar\font45}
\DeclareFontShape{U}{matha}{m}{n}{
      <5> <6> <7> <8> <9> <10> gen * matha
      <10.95> matha10 <12> <14.4> <17.28> <20.74> <24.88> matha12
      }{}
\DeclareSymbolFont{matha}{U}{matha}{m}{n}
\DeclareFontSubstitution{U}{matha}{m}{n}

\DeclareFontFamily{U}{mathx}{\hyphenchar\font45}
\DeclareFontShape{U}{mathx}{m}{n}{
      <5> <6> <7> <8> <9> <10>
      <10.95> <12> <14.4> <17.28> <20.74> <24.88>
      mathx10
      }{}
\DeclareSymbolFont{mathx}{U}{mathx}{m}{n}
\DeclareFontSubstitution{U}{mathx}{m}{n}

\DeclareMathDelimiter{\vvvert}{0}{matha}{"7E}{mathx}{"17}
\DeclarePairedDelimiterX\vertiii[1]{\vvvert}{\vvvert}{\ifbcdot{#1}}

\DeclarePairedDelimiterXPP\trace[1]{\operatorname{Tr}}{(}{)}{}{\ifbcdot{#1}} 
\DeclarePairedDelimiterXPP\col[1]{\operatorname{col}}{\{}{\}}{}{\ifbcdot{#1}} 
\DeclarePairedDelimiterXPP\row[1]{\operatorname{row}}{\{}{\}}{}{\ifbcdot{#1}} 
\DeclarePairedDelimiterXPP\erf[1]{\operatorname{erf}}{(}{)}{}{\ifbcdot{#1}}
\DeclarePairedDelimiterXPP\erfc[1]{\operatorname{erfc}}{(}{)}{}{\ifbcdot{#1}}
\DeclarePairedDelimiterXPP\KLD[2]{D}{(}{)}{}{\ifbcdot{#1}\, \delimsize\|\, \ifbcdot{#2}} 
\DeclarePairedDelimiterXPP\op[2]{\operatorname{#1}}{(}{)}{}{#2} 

\newcommand{\ud}{\,\mathrm{d}} 

\DeclarePairedDelimiterXPP\indicate[1]{{\bf 1}}{\{}{\}}{}{\ifbcdot{#1}}

\NewDocumentCommand\ofrac{s m}{%
	\IfBooleanTF#1%
	{\dfrac{1}{#2}}%
	{\frac{1}{#2}}%
}
\NewDocumentCommand\ddfrac{s m m}{%
	\IfBooleanTF#1%
	{\dfrac{\mathrm{d} {#2}}{\mathrm{d} {#3}}}%
	{\frac{\mathrm{d} {#2}}{\mathrm{d} {#3}}}%
}
\NewDocumentCommand\ppfrac{s m m}{%
	\IfBooleanTF#1%
	{\dfrac{\partial {#2}}{\partial {#3}}}%
	{\frac{\partial {#2}}{\partial {#3}}}%
}

\newcommand{\setgiven}{:}
\providecommand\given{}

\DeclarePairedDelimiterX\Set[2]\{\}{%
	\if#1:%
		\renewcommand\given{\SetSymbol{:}}%
	\else%
		\renewcommand\given{\SetSymbol[\delimsize]{#1}}%
	\fi%
#2
}

\NewDocumentCommand\set{s O{\setgiven} m}{%
	\IfBooleanTF#1%
	{\Set*{#2}{#3}}%
	{\Set{#2}{#3}}%
}

\NewDocumentCommand{\evalat}{ s O{\big} m e{_^} }{%
\IfBooleanTF{#1}%
{\left. #3 \right|}{#3#2|}%
\IfValueT{#4}{_{#4}}%
\IfValueT{#5}{^{#5}}%
}

\providecommand\given{}
\DeclarePairedDelimiterXPP\cprob[1]{}(){}{
\renewcommand\given{\nonscript\,\delimsize\vert\allowbreak\nonscript\,\mathopen{}}%
#1%
}
\DeclarePairedDelimiterXPP\cexp[1]{}[]{}{
\renewcommand\given{\nonscript\,\delimsize\vert\allowbreak\nonscript\,\mathopen{}}%
#1%
}

\DeclareDocumentCommand \P { s e{_^} d() g } {%
	\mathbb{P}%
	\IfBooleanTF{#1}%
		{
			\IfValueT{#2}{_{#2}}%
			\IfValueT{#3}{^{#3}}%
			\IfValueTF{#5}{\cprob{#4 \given #5}}{\IfValueT{#4}{\cprob{#4}}}%
		}%
		{
			\IfValueT{#2}{_{#2}}%
			\IfValueT{#3}{^{#3}}%
			\IfValueTF{#5}{\cprob*{#4 \given #5}}{\IfValueT{#4}{\cprob*{#4}}}%
		}%
}

\DeclareDocumentCommand \E { s e{_^} o g } {%
	\mathbb{E}%
	\IfBooleanTF{#1}%
		{
			\IfValueT{#2}{_{#2}}%
			\IfValueT{#3}{^{#3}}%
			\IfValueTF{#5}{\cexp{#4 \given #5}}{\IfValueT{#4}{\cexp{#4}}}%
		}%
		{
			\IfValueT{#2}{_{#2}}%
			\IfValueT{#3}{^{#3}}%
			\IfValueTF{#5}{\cexp*{#4 \given #5}}{\IfValueT{#4}{\cexp*{#4}}}%
		}%
}

\DeclareDocumentCommand \Var { s e{_^} d() g } {%
	\var%
	\IfBooleanTF{#1}%
		{
			\IfValueT{#2}{_{#2}}%
			\IfValueT{#3}{^{#3}}%
			\IfValueTF{#5}{\cprob{#4 \given #5}}{\IfValueT{#4}{\cprob{#4}}}%
		}%
		{
			\IfValueT{#2}{_{#2}}%
			\IfValueT{#3}{^{#3}}%
			\IfValueTF{#5}{\cprob*{#4 \given #5}}{\IfValueT{#4}{\cprob*{#4}}}%
		}%
}

\DeclareDocumentCommand \Cov { s e{_^} d() g } {%
	\cov%
	\IfBooleanTF{#1}%
		{
			\IfValueT{#2}{_{#2}}%
			\IfValueT{#3}{^{#3}}%
			\IfValueTF{#5}{\cprob{#4 \given #5}}{\IfValueT{#4}{\cprob{#4}}}%
		}%
		{
			\IfValueT{#2}{_{#2}}%
			\IfValueT{#3}{^{#3}}%
			\IfValueTF{#5}{\cprob*{#4 \given #5}}{\IfValueT{#4}{\cprob*{#4}}}%
		}%
}

\ExplSyntaxOn
\NewDocumentCommand \dist {m o o} {%
\mathrm{#1}\left(%
	\IfValueT{#3}{%
		\tl_if_blank:nTF{ #3 }{\cdot\, \middle|\, }{#3\, \middle|\, }%
	}
	\IfValueT{#2}{#2}%
\right)%
}
\ExplSyntaxOff

\NewDocumentCommand {\cbrace} {t+ D[]{black} D(){\widthof{#5}} m m } {%
	\begingroup%
		\color{#2}
		\IfBooleanTF{#1}{%
			\overbrace{#4}^%
		}{
			\underbrace{#4}_%
		}%
		{\parbox[c]{#3}{\centering\footnotesize{#5}}}%
	\endgroup%
}

\let\oldforall\forall
\renewcommand{\forall}{\oldforall \, }

\let\oldexist\exists
\renewcommand{\exists}{\oldexist \, }

\makeatletter

\newcommand{\rankcolor}[2]{%
	\expandafter\renewcommand\csname #1\endcsname[1]{%
		\ifblank{##1}{%
			{\color{#2} \textbf{#2}}%
		}{%
			\ifmmode
				\textcolor{#2}{\bm{##1}}%
			\else%
				{\color{#2} \textbf{##1}}%
			\fi	
		}%
	}
}

\rankcolor{first}{red}
\rankcolor{second}{blue}
\rankcolor{third}{cyan}
\makeatother

\graphicspath{{./Figures/}{./figures/}}
\DeclareDocumentCommand{\includeCroppedPdf}{ o O{./Figures/} m }{
	\IfFileExists{#2#3-crop.pdf}{}{%
		\immediate\write18{pdfcrop #2#3.pdf #2#3-crop.pdf}}%
	\includegraphics[#1]{#2#3-crop.pdf}
}

\makeatletter
\newcommand*{\addFileDependency}[1]{
  \typeout{(#1)}
  \@addtofilelist{#1}
  \IfFileExists{#1}{}{\typeout{No file #1.}}
}
\makeatother

\definecolor{gray90}{gray}{0.9}
\def\colorlist{red,blue,brown,cyan,darkgray,gray,lightgray,green,lime,magenta,olive,orange,pink,purple,teal,violet,white,yellow}

\makeatletter
\def\startcomment{[}
\ifx\nohighlights\undefined
	\newcommand{\createcolor}[1]{%
			\expandafter\newcommand\csname #1\endcsname[1]{{\color{#1} ##1}}%
	}
	\newcommand{\msout}[1]{\text{\color{green} \st{\ensuremath{#1}}}}
	\newcommand{\del}[1]{{\color{green}\ifmmode \msout{#1}\else\st{#1}\fi}}
\else
	\newcommand{\createcolor}[1]{%
			\expandafter\newcommand\csname #1\endcsname[1]{%
				\noexpandarg%
				\StrChar{##1}{1}[\firstletter]%
				\if\firstletter\startcomment%
					\relax
				\else%
					##1
				\fi
			}%
	}
	\newcommand{\msout}[1]{}
	\newcommand{\del}[1]{}
\fi

\def\@tempa#1,{%
    \ifx\relax#1\relax\else
        \createcolor{#1}%
        \expandafter\@tempa
    \fi
}
\expandafter\@tempa\colorlist,\relax,
\makeatother

\newcommand{\hhide}[1]{}

\ifx\diagnoselabel\undefined
	\relax
\else
	\makeatletter
	\def\@testdef #1#2#3{%
		\def\reserved@a{#3}\expandafter \ifx \csname #1@#2\endcsname
			\reserved@a  \else
			\typeout{^^Jlabel #2 changed:^^J%
				\meaning\reserved@a^^J%
				\expandafter\meaning\csname #1@#2\endcsname^^J}%
			\@tempswatrue \fi}
	\makeatother
\fi

\usepackage{amsmath,amsfonts,bm}

\def\eqref#1{equation~\ref{#1}}

\def\1{\bm{1}}

\def\vx{{\bm{x}}}
\def\vy{{\bm{y}}}

\DeclareMathAlphabet{\mathsfit}{\encodingdefault}{\sfdefault}{m}{sl}
\SetMathAlphabet{\mathsfit}{bold}{\encodingdefault}{\sfdefault}{bx}{n}

\def\sM{{\mathbb{M}}}

\def\sS{{\mathbb{S}}}

\newcommand{\R}{\mathbb{R}}

\begin{document}

\title{Edge centrality and the total variation of graph distributional signals }

\author{Feng~Ji 
\thanks{F. Ji is with the School of Electrical and Electronic Engineering, Nanyang Technological University, 639798, Singapore (Email: jifeng@ntu.edu.sg).}
}



\maketitle

\begin{abstract}
This short note is a supplement to \cite{Ji23}, in which the total variation of graph distributional signals is introduced and studied. We introduce a different formulation of total variation and relate it to the notion of edge centrality. The relation provides a different perspective of total variation and may facilitate its computation.     
\end{abstract}


\subsection{Graph distributional signals and notions total variation} The concept of graph distributional signals is introduced and studied in \cite{Ji23, Jif23, Ji2023}. In \cite{Jif23}, signal processing aspects such as Fourier transform and filtering are studied in detail. On the other hand, the notion of total variation of graph distributional signals is defined in \cite{Ji23}, which we recall as follows.

Let $G=(V,E)$ be a finite, connected and unweighted simple graph. Let $\vx = (\vx(i))_{1\leq i\leq n}$ be a graph signal, where $\vx(i)$ belongs to a signal space such as $\mathbb{R}$, $\bbS$ (a finite set with discrete metric). Its (squared) \emph{total variation} (\cite{Shu13}) is defined by 
\begin{align} \label{eq:tvv}
\mathcal{T}_G(\vx) = \sum_{(v_i,v_j) \in E}(\vx(i)-\vx(j))^2. 
\end{align}

To generalize, for a metric space $\sM$ (e.g., $\mathbb{R}^n$, $\bbS^n$), let $\mathcal{P}(\sM)$ be the space of probability measures on $\sM$ (\gls{wrt} the Borel $\sigma$-algebra) having finite second moments. 

\begin{Definition} \label{defn:tms} 
Given a graph $G$ with $n$ nodes and graph signal space $\sM$ such as $\bbS^n$ or $\bbR^n$, an element $\mu \in \mathcal{P}(\sM)$ is called a \emph{graph distributional signal}. The marginals of a graph distributional signal $\mu$ are the marginal distributions $\scN = \set{\mu_i \given 1\leq i\leq n}$ \gls{wrt} the $n$ coordinates of $\sM$. 
\end{Definition}

Given a graph distributional signal $\mu \in \mathcal{P}(\sM)$ for $\sM = \sS^n$ or $\R^n$, its total variation can be modified directly from (\ref{eq:tvv}) as follows:
\begin{align} \label{eq:tmu}
    \mathcal{T}_G(\mu) = \E_{\vx \sim \mu}\mathcal{T}_G(\vx) = \int \sum_{(v_i,v_j)\in E} d(\vx(i),\vx(j))^2 \ud\mu(\vx).
\end{align}
where $d$ is the metric on the signal space. Modeled on the Wasserstein metric $W(\cdot,\cdot)$ \cite{Vil09}, the total variation of a set of graph distributional marginals \cite{Ji23} is as follows.

\begin{Definition}
Given the marginals $\scN = \set{\mu_i \given 1\leq i\leq n}$ of a graph distributional signal in $\calP(\sM)$, the total variation of $\scN$ is
\begin{align}
    \mathcal{T}_G(\scN) = \inf_{\mu \in \Gamma(\scN)} \mathcal{T}_G(\mu),
\end{align}
where $\Gamma(\scN)$ is the collection of all graph distributional signals in $\mathcal{P}(\sM)$ whose marginals agree with $\scN$. 
\end{Definition}

For marginal distributions $\scN = (\nu_v)_{v\in V}$, we have introduced the total variation $\calT_G(\scN)$. However, computing the total variation often requires solving an optimization problem that is usually hard or even impossible to solve. It does not have a closed-form formula. In GSP, a useful strategy to understand a general graph $G$ is to study the spanning trees of $G$ \cite{Sha11, Ji16, Jif19}, when the estimation of a quantity defined on $G$ is intractable. This inspires us to aggregate the total variation of $\scN$ on subtrees of $G$. Naturally, we may consider expectation as a way of aggregation, which leads to the following framework.

\begin{Definition}
Let $\mathbb{T}_G$ be the finite set of all subtrees of $G$. We use $\mathfrak{M}_G$ to denote the space of probability distributions on $\mathbb{T}_G$. For $\eta \in \mathfrak{M}_G$, define the total variation of $\scN$ w.r.t.\ $\eta$ as 
\begin{align} \label{eq:cms}
    \calT_{\eta}(\scN) = \mathbb{E}_{T \sim \eta}\calT_{T}(\scN).
\end{align}
Here, if $T$ does not contain all the nodes of $G$, then $\calT_{T}(\scN)$ considers $\scN$ as the subset restricted to nodes of $T$.
\end{Definition}

Unlike \cite{Ji23}, $\eta$ is fixed and independent of $\scN$. Therefore, $\calT_{\eta}$ is a proxy of $\calT_G$ that is different from that in \cite{Ji23}. We relate it to the notion of ``centrality'' in the next section, which may facilitate its computation. 

\subsection{Connection to edge centralities} \emph{Centrality} is a well-studied topic in network analysis. An \emph{edge centrality} is a function: $C: E \to \mathbb{R}_+$. As $E$ is a finite set, $C$ can also be arranged as a vector $C = \big(C(e)\big)_{e\in E}$. This means we can take the inner product. Let $\mathfrak{C}_G$ be the set of edge centralities on $G$. We have the following result.

\begin{Theorem} \label{thm:tia}
    There is a map $\phi_G: \mathfrak{M}_G \to \mathfrak{C}_G$ such that the following identity holds for any $\eta\in \mathfrak{M}_G$ and marginal $\scN = (\mu_i)_{v_i\in V}$: 
    \begin{align} \label{eq:cam}
        \calT_{\eta}(\scN) = \langle \phi_G(\eta), W_{\scN} \rangle, 
    \end{align}
    where $W_{\scN} = \big(W_{\scN}(e)\big)_{e\in E}$ is the (edge) vector $W_{\scN}(e) = W(\mu_i,\mu_j)^2$ with $e=(v_i,v_j)$, where $W(\cdot, \cdot)$ is the Wasserstein metric. Moreover, $\phi_G$ is unique if $\mathbb{R}^n \subset \mathbb{M}$. 
\end{Theorem}

\begin{proof}
We first claim that for any tree $T=(V_T,E_T)$, we have $\calT_T(\scN) = \sum_{(v_i,v_j)\in E_T} W(\mu_i,\mu_j)^2$.

For any $\epsilon>0$, there is a $\mu_0 \in \Gamma(\scN)$ such that  $\calT_T(\calN) \geq \calT_T(\mu_0) - \epsilon$. For each edge $(v_i,v_j) \in E$, let $\mu_{0,i,j}$ be the marginal distribution of the pair $(v_i,v_j)$. The marginals of $\mu_{0,i,j}$ are $\mu_i$ and $\mu_j$ at $v_i$ and $v_j$ respectively. We have
\begin{align*}
    &\mathcal{T}_T(\scN) -\epsilon  \geq \mathcal{T}_T(\mu_0) \\
    & = \int \sum_{(v_i,v_j)\in E_T} d(\vx(i),\vx(j))^2 \ud\mu_0(\vx) \\
    & = \sum_{(v_i,v_j) \in E_T} \int d(\vx(i),\vx(j))^2 \ud\mu_0(\vx) \\
    &= \sum_{(v_i,v_j) \in E_T} \int d(\vy(i),\vy(j))^2 \ud\mu_{0,i,j}(\vy) \\
    & \geq \sum_{(v_i,v_j)\in E_T} W(\mu_i,\mu_j)^2.
\end{align*}
Letting $\epsilon \to 0$, we see $\mathcal{T}_T(\scN) \geq  \sum_{(v_i,v_j)\in E_T} W(\mu_i,\mu_j)^2$.

In the opposite direction, fix a node $v_0 \in T$ as the source and orient all the edges away from $v_0$. For each edge $e=(v_i,v_j)$, let $\mu_e \in \Gamma(\{\mu_i,\mu_j\})$ that realizes $W(\mu_i,\mu_j)^2$. Using the orientation $T$ and Bayesian network \cite{Bis06} Section 8.1, there is a joint distribution $\mu_0$ on all the nodes whose marginal of each $e=(v_i,v_j)$ is $\mu_e$. Therefore, similar to the estimations above, we have
\begin{align*}
    & \sum_{(v_i,v_j)\in E_T} W(\mu_i,\mu_j)^2 \\
    & = \sum_{e=(v_i,v_j) \in E_T} \int d(\vy(i),\vy(j))^2 \ud\mu_e(\vy) \\
    & = \calT_T(\mu_0) \geq \calT_T(\scN).
\end{align*}
The claim is proved. 

For a general $G$ and $\eta = \big(p(T)\big)_{T\in \mathbb{T}_G} \in \mathfrak{M}_G$, we may compute that 
\begin{align*}
    & \calT_{\eta}(\scN) = \sum_{T \in \mathbb{T}_G}p(T)\calT_{T}(\scN) \\
    & = \sum_{T \in \mathbb{T}_G}p(T)\sum_{e=(v_i,v_j)\in T}W_{\scN}(e) \\
    & = \sum_{e=(v_i,v_j)\in E} W_{\scN}(e) \sum_{e\in T} p(T) \\
    & = \sum_{e=(v_i,v_j)\in E} W_{\scN}(e) \cdot \big(\mathbb{E}_{T\sim \eta} \chi_e \big),
\end{align*}
where $\chi_e: \mathbb{T}_G \to \{0,1\}$ is the indicator function, i.e, $\chi_e(T)=1$ if $e\in T$ and $0$ otherwise. Hence, to satisfy (\ref{eq:cam}), it suffices to let the centrality $C_{\eta} = \phi_G(\eta)$ to be defined by 
\begin{align} \label{eq:vbx}
C_{\eta}(e) = \mathbb{E}_{T\sim \eta} \chi_e.
\end{align}

Assuming $\mathbb{R}^n \subset \mathbb{M}$, to show $\phi_G$ is unique, it suffices to construct $m=|E|$ ordinary graph signals $\vx_1,\ldots, \vx_m$ such that the following holds. Write $\vx_i = (x_{i,j})_{1\leq j\leq n}$ as a vector and we form 
\begin{align*} 
\vy_i = \big( (x_{i,k}-x_{i,l})^2 \big)_{(v_k,v_l) \in E}.
\end{align*}
It suffices to require that the $m$-dimensional vectors $(\vy_i)_{1\leq i\leq m}$ are linearly independent, i.e., the $m\times m$ matrix $\bY$ with rows $(\vy_i)_{1\leq i\leq m}$ has non-zero determinant. However, if we view $X = (x_{i,j})_{\leq i\leq m, 1\leq j\leq n}$ as independent variables, then $\det(\bY)$ is a nonzero polynomial in $X$. Therefore, by Hilbert's Nullstellensatz, we can always find numbers to substitute the variables in $X$ such that the resulting $\det(\bY)$ is nonzero. The claim is proved.    
\end{proof}

The result essentially claims that there is a correspondence between edge centralities and subtree distributions. Therefore, to study such a distribution and the associated total variation, one may consider the corresponding edge centrality, and vice versa. 

\subsection{Examples} There are many choices of $\eta$ that give rise to familiar centralities $\phi_G(\eta)$ (see below), which can always be computed efficiently. This means that, by the theorem, we can compute the corresponding total variation efficiently. 

Suppose two edge centralities $C_1, C_2$ differ by scaling, i.e., there is $r>0$ such that $C_1 = rC_2$. We usually treat them as the same in providing information such as edge importance. We may thus define an equivalence relation ``$\sim$'' on $\mathfrak{C}_G$ as $C_1 \sim C_2$ if they differ by scaling. From the expression (\ref{eq:vbx}), we see that $\phi_G$ is linear. Since $\mathfrak{M}_G$ is compact, so is $\ima \phi_G$. However, the equivalence class of $\ima \phi$ can contain many more edge centralities. For any graph $|\mathbb{T}_G|$ is usually much larger than $|E|$. Hence, $\phi_G$ is rarely injective. Hence, given an edge centrality $C$ in $\ima \phi$, it usually requires additional prior information to recover or estimate the distribution that gives rise to $C$, such as the support and smoothness.  

We next give some examples. In all the examples, we write $C_{\eta}$ for $\phi_G(\eta)$. 

\begin{Example}
    \begin{enumerate}[(a)]
    \item Constant centrality: $\eta \in \mathfrak{M}_G$ is supported on all the edges with equal probability for each edge. In this case, $C_{\eta}(e) = 1/|E|$ and the resulting total variation is equivalent to the ordinary total variation (for ordinary graph signals) in GSP.
        \item Betweenness centrality: $\eta \in \mathfrak{M}_G$ is supported on the set of geodesic paths $G$ constructed in the following two-step sampling procedure. We uniformly sample a pair of distinct nodes $v_i,v_j$. A geodesic $P_{v_i,v_j}$ between $v_i,v_j$ is uniformly sampled among all the geodesic between $v_i,v_j$. Hence, 
        \begin{align*}
        C_{\eta}(e) = \frac{1}{c}\sum_{v_i\neq v_j \in V}\frac{\# P_{v_i,v_j,e}}{\# P_{v_i,v_j}},
        \end{align*}
        where $P_{v_i,v_j,e}$ is a geodesic path between $v_i$ and $v_j$ containing $e$. The constant $c$ is the number of distinct pairs of nodes. It is apparent that $C_{\eta}$ is equivalent to the betweenness centrality \cite{Bra01}. 
        \item Spanning tree centrality: $\eta \in \mathfrak{M}_G$ is the uniform distribution supported on the set of \emph{spanning trees} of $G$. Then for each $e$, $C_{\eta}(e)$ is the percentage of spanning trees containing $e$. This is nothing but the spanning tree centrality \cite{Mav15}.  
    \end{enumerate}
\end{Example}

By Theorem~\ref{thm:tia}, we have closed form formula for $\calT_{\eta}(\scN)$ if we have that for $W(\mu_i,\mu_j)$. We conclude this note by showing examples for which the above holds. 

\begin{Example}
    \begin{enumerate}[(a)]
        \item Suppose $\mathbb{M} = \mathbb{R}^n$ and $\mu_i$ is the normal distribution on $\mathbb{R}$ with mean $m_i$ and standard deviation $s_i$. Then $W(\mu_i,\mu_j)^2 = (m_i-m_j)^2 + (s_i-s_j)^2$. 
        \item Suppose $\mathbb{M} = \mathbb{R}^n$ and $\mu_i$ is the empirical distribution supported on the (increasingly) ordered set $\{x_{i,1},\ldots,x_{i,N}\}$, where $N$ is independent of $i$. Then 
        \begin{align*}
        W(\mu_i,\mu_j)^2 = \frac{1}{N}\sum_{1\leq k\leq N}(x_{i,k}-x_{j,k})^2.
        \end{align*}
        \item Suppose $\mathbb{M} = \mathbb{S}^n$, where $\mathbb{S} = \{s_1,\ldots, s_N\}$ is a finite discrete metric space. For $1\leq k\leq N$, let $p_{i,k}$ be the probability weight of $s_k$ for $\mu_i$. Then 
        \begin{align*}
            W(\mu_i,\mu_j)^2 = \frac{1}{2} \sum_{1\leq k\leq N} |p_{i,k}-p_{j,k}|. 
        \end{align*}
    \end{enumerate}
\end{Example}

\bibliographystyle{IEEEtran}
\bibliography{IEEEabrv,StringDefinitions,tsipn_bib}

\end{document}